\documentclass[aoas,preprint]{imsart}

\RequirePackage[OT1]{fontenc}
\RequirePackage{graphicx}
\RequirePackage{amsthm}
\RequirePackage[cmex10]{amsmath}
\RequirePackage{natbib}
\RequirePackage[colorlinks,citecolor=blue,urlcolor=blue]{hyperref}
\RequirePackage{hypernat}

\RequirePackage[OT1]{fontenc}
\RequirePackage{amsthm,amsmath}
\usepackage{bm,graphicx,caption,subcaption,color}

\startlocaldefs
\numberwithin{equation}{section}
\theoremstyle{plain}

\endlocaldefs

\newcommand{\bfun}{\left\{\begin{array}{ll}}
\newcommand{\efun}{\end{array}\right.}

\newtheorem{lemma}{Lemma}

\startlocaldefs
\newcommand{\comment}[1]{}
\numberwithin{equation}{section}
\theoremstyle{plain}
\endlocaldefs

\begin{document}

\begin{frontmatter}
\title{Hierarchical Semi-parametric Duration Models}
\runtitle{Hierarchical Semi-parametric Duration Models}

\begin{aug}
\author{\fnms{Mingyu} \snm{Tang} \ead[label=e1]{mingyut@cmu.edu}}
\and
\author{\fnms{Mark} \snm{Schervish} \ead[label=e2]{mark@cmu.edu}
}

\affiliation{Department of Statistics, Carnegie Mellon University}

\address{Baker Hall 232, 5000 Forbes Avenue, Pittsburgh, PA, 15213\\
\printead{e1}\\
\phantom{E-mail:\ }\printead*{e2}}

\end{aug}

\begin{abstract}
This research attempts to model the stochastic process
of trades in a limit order book market as
a marked point process.  We propose a semi-parametric model for the
conditional distribution given the past, attempting to capture the effect of
the recent past in a nonparametric way and the effect of the more
distant past using a parametric time series model. Our framework
provides more flexibility than the most commonly used family of
models, known as Autoregressive Conditional Duration (ACD), in terms
of the shape of the density of durations and in the form of dependence
across time. We also propose an online learning algorithm for intraday
trends that vary from day to day.  This allows us both to do prediction of
future trade times and to incorporate the effects of
additional explanatory variables. In this paper, we show that the framework works better
than the ACD family both in the sense of prediction log-likelihood and
according to various diagnostic tests using data from the New York
Stock Exchange. In general, the framework can be used both to estimate
the intensity of a point process, and to estimate a the joint density
of a time series.
\end{abstract}

\begin{keyword}
\kwd{microstructure}
\kwd{duration model}
\kwd{semiparametric}
\end{keyword}

\end{frontmatter}

\section{Introduction}

In today's financial world, most markets rely on a {\em limit order
  book} (LOB) to match buyers and sellers. High frequency traders and
market makers rely on real-time access to the LOB in order to
implement their trading algorithms. 
Researchers, traders and regulators are interested in the dynamics of
LOB for their own individual reasons.
How stocks trade on the NYSE is discussed by Schwartz (1993) 
\cite{NYSE1} and Hasbrouck, Sofianos and Sosebee (1993) \cite{NYSE2}.

A common feature of the dynamics of each LOB  is clustering of events.
Each time that an event occurs, the likelihood increases that another
event will occur in the near figure.  Processes with this feature are
called {\em self-exciting}, and the lengths of consecutive durations
(times between events) tend to be similar to each other.  Two main
types of models for 
self-exciting processes have been developed over the years.
One type consists of {\em intensity models}, and the other consists of
{\em duration models}. Intensity models focus on the conditional intensity
function, which gives the instantaneous conditional probability of an event
at each time given the history of the
process. Bauwens and Hautsch (2006) \cite{ABC5} give a good survey of
current intensity models. The most common intensity
model is the Hawkes process, introduced by Hawkes (1971)
\cite{HAWKES}, in which the conditional intensity function is modeled
as a linear combination of the effects of all of the past events. The
effect of each past event is modeled as a decaying function
of the time elapsed since that event. Zhao's thesis
(2010) \cite{ABC6} proposed another intensity model, in which the
conditional intensity function depends on the number of events in the
most recent 
past time window of a fixed length. Duration models were
popularized by Engle and Russel (1997, 1998) \cite{ACD} who introduced the
{\em autoregressive conditional duration} (ACD) model.  The ACD model
specifies that the conditional mean of the next duration has an
autoregressive structure 
so as to capture the self-exciting effect. There is rich literature
extending the ACD model to other duration models. Pacurar (2006)
\cite{ACDsurvey} has given a comprehensive survey on the development 
of ACD models. In particular, Bauwens and Veredas (2004) \cite{SCD}
proposed a stochastic conditional duration (SCD) model, in which they
introduce a stochastic noise in the autoregressive formula for the
expectation of duration. 
Also, the log ACD model \cite{logACD}, the threshold ACD model
\cite{TACD}, the Markov Switching ACD model \cite{MSACD}, the
stochastic volatility duration (SVD) model \cite{SVD}, and the fractionally
integrated ACD model \cite{FIACD} are all designed to generalize and
improve the original ACD model in various ways.  
Furthermore, Russell (1999) \cite{ACI} has developed an intensity model 
based on the ACD which is called the autoregressive conditional
intensity (ACI) model.

Although models in the ACD family successfully capture the
self-exciting features of the duration processes, Bauwens, Giot,
Grammig and Veredas (2000) \cite{COMPARE} have shown that none of the
parametric ACD models can pass a model evaluation criterion that was
proposed by Diebold, Gunther and Tay (1998) \cite{denFORCAST} and is based on
the probability integral transform theorem. The criterion will be
described in Section~\ref{subsec:diagnostics}. In addition, all of the
models incorporate an intraday trend.  Typically, a spline is fit to
the durations as a preprocessing step,
and then the durations are divided by the fitted
spline to remove the trend.  In order to  predict durations
on one day given what one learns from the previous day, it is useful
to have a model for how the intraday trend varies from day to day.

In this article, we develop a semiparametric duration model for the dynamics of
trade flow.  We combine nonparametric
conditional density estimation, parametric time series models and an
online learning of the intraday trend.  Since trades occur at 
irregular times throughout the day, the trade duration process is
typically characterized as a marked point process, where the trades
are target events and the other associated features, including price,
spread, trade side, and other features of the LOB are marks. 
The main contributions of this article lie in a new and more precise
model of the dynamics of the trade duration process as well as a new
way to deal with the intraday trend for the purposes of prediction.

The rest of the paper is organized as follows. In
Section~\ref{sec:acd}, we review the ACD model and some of its variants. 
Section~\ref{sec:hsdm}, introduces our semiparametric framework for
modeling the duration process, including our estimation procedures.
Section~\ref{sec:expres} presents experimental results
on LOB data  from the NYSE and compares our
framework with  four models from the ACD family. 
Section~\ref{sec:disc} summarizes our results and puts them into perspective.


\section{Review of ACD models}\label{sec:acd}
In this section, we briefly review the ACD model along with some of its
variants and analyze their limitations.  These limitations serve as
the inspiration for our semiparametric model.  
Let $X_i$ denote the elapsed time (duration) between two consecutive
events at times $time_{i-1}$ and $time_i$, i.e. $X_i = time_i - time_{i-1}$, with
$time_0$ being the time at which observation begins.
An ACD model attempts to capture the time dependence in the duration
process by modeling the conditional expectation of the next duration
given the past, i.e. $E(X_i | \mathcal{F}_{i-1})$, where $\mathcal{F}_{i-1}$
denotes the information available up to time $time_{i-1}$. A common ACD model
is:
\begin{eqnarray} 
\label{eqn:acd1}
X_i& =& \Psi_i \epsilon_i,\\
\label{eqn:acd2}
\Psi_i &=& \omega + \alpha X_{i-1} + \beta \Psi_{i-1},
\end{eqnarray} 
where $\{\epsilon_i\}_{i\geq1}$ is a process of IID positive random
variables with mean 1, $\omega>0$, $\alpha>0$ and $\beta>0$ are
parameters with $\alpha+\beta<1$ (to allow the $\Psi_i$ to have a
common mean.)  Thus, $E(X_i | \mathcal{F}_{i-1}) = \Psi_i$.
The particular model specified above is called ACD(1,1) because of the
introduction of one lag for both $X$ and $\Psi$ in (\ref{eqn:acd2}).
The distribution of $\epsilon_i$ is assumed to be from a parametric
family with a long tail. Common choices include, Gamma, Weibull and
Burr families. 

\subsection{Additional Explanatory Variables in ACD}

In a market microstructure data set, such as our NYSE data set, events
are usually associated with some additional explanatory variables,
such as volume, spread, price and so forth. These variables can be
characterized as marks in the point process and they can have an
impact on the intensity of the process. In the original ACD model and its
variants, the effects of additional explanatory variables are
incorporated by modifying the autoregressive formula (\ref{eqn:acd2}) to
\begin{equation}\label{eqn:AEV}
\Psi_i = \omega + \alpha X_{i-1} + \beta \Psi_{i-1} + \delta^T u_i,
\end{equation}
where, $u_i$ is a vector of additional explanatory variables and
$\delta$ is a vector of coefficients. Such a specification
indicates that the additional variables affect the distribution of
durations by a scale change.

\subsection{ACD Variants}
Researchers have created variations of the ACD model of two
main types.  One type of variation modifies the autoregressive formula
(\ref{eqn:acd2}). The  Log-ACD model by 
Bauwens and Giot (2000) \cite{logACD} replaces (\ref{eqn:acd2}) by
\[\log(\Psi_i) = \omega + \alpha\log(X_{i-1}) + \beta\log(\Psi_{i-1}),\]
which, unlike (\ref{eqn:acd2}), requires no additional restrictions in
order to guarantee that $\Psi_i>0$.  The Stochastic
Duration model \cite{SCD} introduces a random noise in  (\ref{eqn:acd2})
to allow $\Psi_i$ to be a non-deterministic function of the past, as
in
\[\Psi_i = \omega + \alpha X_{i-1} + \beta \Psi_{i-1} + u_i,\]
where $u_i$ is IID Gaussian noise. The Fractional Integrated ACD model
\cite{FIACD} introduced a differencing in order to capture the long
memory of the duration sequence.  This model will be described in more
detail in Section~\ref{sec:bench} because we use it as a benchmark for
comparison to our model.  The Threshold ACD \cite{TACD},
\[   \Psi_i = \left\{ 
     \begin{array}{lr} 
       \omega_1 + \alpha_1 X_{i-1} + \beta_1 \Psi_{i-1}, &
       \mbox{if \ $0 < X_{i-1} \leq r_1$,} \\ 
       \omega_2 + \alpha_2 X_{i-1} + \beta_2 \Psi_{i-1}, &
       \mbox{if $r_1 < X_{i-1} \leq r_2$,} \\ 
       \omega_3 + \alpha_3 X_{i-1} + \beta_3 \Psi_{i-1},
       &\mbox{if $r_2 < X_{i-1} < \infty$.}\end{array}\right.,\] 
allows different dependence in different
regimes. 

The second type of variation is to allow more general
distributions for $\epsilon$ in (\ref{eqn:acd1}). Traditionally, $\epsilon$
is assumed to have a long-tailed distribution with mean 1.  A semiparametric
version of the ACD model \cite{FIACD} estimates the parameters in
(\ref{eqn:acd2}) by quasi-maximum likelihood \cite{QML} and then
estimates the distribution of the residual $\epsilon$
nonparametrically. When we compare our model with
ACD variants, we will include both the parametric ACD with exponential
distribution and the semiparametric ACD using kernel density estimation to
estimate the distribution of $\epsilon$ nonparametrically.  

\subsection{Limitations of ACD Models}\label{sec:limit}
In parametric ACD models and their variants, there is a strong parametric
assumption on the distribution of the $\epsilon$ process. Gamma and
Weibull distributions are the most common choices, and both of these
include exponential distributions as special cases. However, as we
will show in Section~\ref{sec:hsdm}, the trade durations do not
admit such an 
ideal parametric distribution. Furthermore, the empirical distribution
of log-durations appears to be bimodal, which undermines the
performance of any model that relies on a parametric family of
distributions for log-durations. In
  nonparametric versions of the ACD model, Gaussian kernel
  density estimation also performs poorly because of the long tail of
  the distribution of the residual $\epsilon$. 

Another important restriction on ACD models and their variants is that
the time dependency is incorporated only in the expectation (which
happens to be the same as the scale because of the form of the model) of the
duration distribution. However, as we will show in our experimental
results in Section~\ref{sec:expres},
the previous trade duration affects the ensuing trade duration in
a more general way. In particular, it changes both the locations and
the relative sizes of the two modes of the bimodal conditional
distribution of durations. Therefore, modeling the time 
dependency solely in terms of the duration mean/scale cannot capture the
more general effect of previous durations. Although some of the extensions
are designed to overcome this shortcoming,  they are still not
flexible enough to capture the effects on the modes of the distribution.


\section{ Hierarchical Semi-parametric Duration Model
  (HSDM)} \label{sec:hsdm} 
In this section, we propose a semiparametric model for the
conditional distribution of durations.  The model also allows
estimation of a conditional intensity function.  Suppose that we
observe a series 
of occurrence times $\{time_i\}_{i\geq0}$ from a point process. By  taking the
differences of successive occurrence times, we get a series of
durations $\{\Delta time_i\}_{i\geq1}$, where $\Delta time_i=time_i-time_{i-1}$.
Suppose that the conditional density 
function (given the past) for $\Delta time_i$ is $p_i(t)$ with CDF 
$P_i(t)$. From the viewpoint of point processes, the conditional
intensity function $\lambda$ is defined as: 
\[\lambda(t | \mathcal{F}_t) = \lim_{\Delta t\rightarrow 0}
\frac{1}{\Delta t} \Pr(\mbox{One\ event\ occurs\ in\ the\ time\
  interval} \ [t, t+\Delta t] | \mathcal{F}_t)\] 
In a duration-based point process, the estimated intensity $\hat
\lambda_t$ is essentially the hazard function, which can be obtained
by $\hat haz_i(\cdot) = \hat p_i(\cdot) / (1 - \hat P_i(\cdot) )$. 

Since the duration has a long tail, we transform to the log scale in
the rest of the article and use $\{T_i\}_{i\geq1}$ to denote the logarithms of
the durations.  The logarithms of durations are critical in that
  the original durations have a distribution with a long tail, which
  makes Gaussian kernel density estimation perform poorly. Instead,
 kernel density estimation on the log scale is
  equivalent to  kernel density estimation with an asymmetric
  bandwidth that increases the farther one
  gets into the tail of the original scale.  This ensures that the asymmetric
  long-tailed distribution is captured well. Throughout the paper, the
  analysis of HSDM focuses on logarithms of durations.

We let $f_i(\cdot)$ denote the density of $T_i$ given the past, and
$F_i(\cdot)$ denotes its 
CDF. Naturally, $p_i(\cdot)$ and $P_i(\cdot)$ can be derived directly from
$f_i(\cdot)$ and $F_i(\cdot)$ and vice-versa. Therefore, our goal is
equivalent to estimating the density function 
$f_i(\cdot)$ of $T_i $ given $\{T_j\}_{j=1}^{i-1}$.  
In what follows, we present the model, a corresponding estimation method, and a
prediction algorithm for future events. 

\subsection{Model} \label{sec:DGP}
In this section, we describe a model for general point processes. In
Section~\ref{subsec:estimation}, 
we give the specific version that we use for trade duration processes
along with the steps needed to fit the model.
In general, the duration sequence comes from a multi-layer
hierarchical model as shown in Figure~\ref{fig:DGP1},

\begin{figure}[htb]
\begin{center}
\begin{subfigure}{0.4\textwidth}
\caption{}
\includegraphics[width = \textwidth, height = 3.0in ]{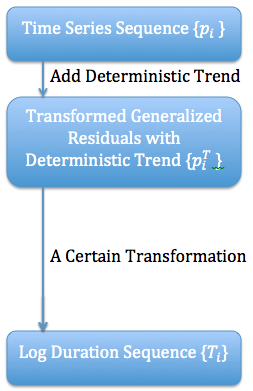}
\label{fig:DGP1}
\end{subfigure}
~
\begin{subfigure}{0.4\textwidth}
\caption{}
\includegraphics[width = \textwidth, height = 3.0in]{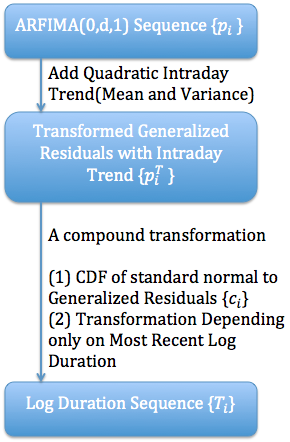}
\label{fig:DGP2}
\end{subfigure}
\end{center}
\footnotesize
\caption{(a) General form of the model, and (b) specific version for trade
  durations.}
\label{fig:DGP}
\end{figure}

with Figure~\ref{fig:DGP2} showing more specific information for the
model that we use for trade durations. The process that we used for
choosing the components of the model is described in
Section~\ref{sec:modelbuild}. 

\begin{enumerate}
\item Underlying the process is a latent process $\{p_i\}_{i\geq1}$
  that captures the long-memory dependence.  The latent process is
  modeled by a parametric time series model, e.g.  autoregressive
  moving average (ARMA) or autoregressive fractionally integrated
  moving average  (ARFIMA). If one needs to incorporate
additional explanatory variables, one can
augment the time series model with a regression component.  Descriptions of
ARFIMA models and the augmentation that we use for regression are is given in
Appendix~\ref{sec:arfima}. 
\item A general time trend can be modeled so that the distribution of 
duration depends on both clock time and event time.  Let
$R(\cdot|time_{i-1})$ be a one-to-one, clock-time dependent, function of a
real variable, where $time_{i-1}$ stands for the clock time at which event $i-1$
occurs. Applying the transformation produces the transformed process 
\[p^T_i = R(p_i|time_{i-1}).\]
\item Finally, the log-duration $T_i$ is a general past-dependent
  transformation of $p^T_i$, $T_i=H_i(p^T_i)$. 
\end{enumerate}
Combining the above levels of the hierarchy, the distribution of
$\{T_i\}_{i\geq1}$ can be written in terms of the distribution of the
latent process $\{p_i\}_{i\geq1}$.  Let $G_i(\cdot)$ denote the
conditional distribution function of $p_i$ given the past.  Then, the
conditional distribution of $T_i$ given the past has CDF
\begin{equation}\label{eq:tcdf}
F_i^*(t)=\Pr(T_i\leq t)=\Pr(p_i^T\leq
H_i^{-1}(t))=G_i\left(R^{-1}[H_i^{-1}(t)|time_{i-1}]\right). 
\end{equation} 

\subsection{Estimation Method} \label{subsec:estimation}
Our proposed semiparametric estimation method proceeds by reversing the
steps in the  data generating process described above.

\begin{enumerate} 
\item First, express the general transformation $H_i$ as
  $H_i(\cdot)=F_i^{-1}(\Phi(\cdot))$, where $F_i(\cdot)$ is a general
  cumulative distribution function (CDF) that depends on the past and
  $\Phi$ is the standard normal CDF.  For trade durations, we
 find a nonparametric kernel estimator of
  $H_i$ as follows.  Compute a  conditional density estimator
  $\hat f_i (\cdot)$ for the log-durations $T_i$ given the previous
  log-duration $T_{i-1}$, and form the
  corresponding CDF, $\hat F_i (\cdot)=\hat F(\cdot|T_{i-1})$. \label{step1} 
Calculate the {\em generalized residuals} $c_i = \hat F_i (T_i)$
  and the {\em transformed generalized residuals} $\hat{p}^T_i =
  \Phi^{-1}(c_i)$. 
\item\label{step2} Model the clock-time dependent trend.  For trade
  durations, we find that an intraday trend is both useful and
  meaningful.  The specific form we use is
\[R(p_i|time_{i-1})=p_i\tau_{\rm sd}(time_{i-1}) + \tau_{\rm mean}(time_{i-1}),\]
where $\tau_{\rm mean}$ and  $\tau_{\rm sd}$ are respectively
functions of clock time that model changes in the mean and 
  standard deviation of 
  $\{\hat{p}^T_i\}_{i\geq1}$.
Estimate the trends  as  $\hat\tau_{\rm mean}(time_{i-1})$ and $\hat\tau_{\rm
    sd}(time_{i-1})$, and then
  calculate the detrended sequence 
\[\hat{p}_i = \hat{R}^{-1}[\hat{p}^T_i|time_{i-1}]=\frac{\hat{p}^T_i-
\hat{\tau}_{\rm mean}(time_{i-1})}{\hat{\tau}_{\rm
    sd}(time_{i-1})}.\]
Specifically, we let both $\tau_{\rm mean}$ and $\tau^2_{\rm sd}$ be
quadratic functions of clock time as  described in more detail below.
\item Fit a parametric time series model to the detrended transformed
  generalized residuals: $TSmodel(\hat{p}_i |\{\hat{p}_j\}_{j=1}^{i-1})$. 
The fitted time series model will predict that each $\hat p_i$ given the
past has a normal distribution with mean $\hat\mu_i$ and standard
deviation $\hat\sigma_i$.
 If additional explanatory variables are needed, an
  appropriate modification is done at this
  step. Section~\ref{sec:arfima} gives more details.  \label{step3} 
\end{enumerate} 
After the fit, transform back to the log-duration scale.
The fitted value corresponding to the $i$th log-duration is
\[\hat{T}_i=\hat
F_i^{-1}\left[\Phi\left(\frac{\hat{p}_i-\hat{\mu}_i}{\hat{\sigma}_i}
\right)\right].\] 
In step~\ref{step1}, $\hat{F}_i$ can be any
past-dependent CDF.  For trade durations, we try to capture
a general form of the most important dependence, specifically the dependence
on the previous log-duration, $T_{i-1}$.  So, we use a nonparametric
conditional density estimator for the density of $T_i$ given
$T_{i-1}$, and convert the density into its corresponding CDF.
If $\hat{F}_i$ were indeed the conditional CDF of
$T_i$ given the past, then the generalized residuals
$\{c_i\}_{i\geq1}$ would be independent 
uniform random variables on the interval $(0,1)$, and
$\hat{p}^T_i=\Phi^{-1}(c_i)$ would be independent standard normal
random variables.  Of course, empirical evidence with trade durations
suggests that the distribution of $\{\hat{p}^T_i\}_{i\geq1}$ is much
more complicated, having both time-varying mean and time-varying
standard deviation, not to mention long memory.  

In step~\ref{step2}, we detrend the $\hat{p}^T_i$.  Both the mean and
variance appear to be large in the middle of the day and smaller at
the start and end of the day. So, we fit a quadratic trend for each:
\[\tau_{\rm mean}(t)=at^2+bt+c,\quad\tau^2_{\rm sd}(t)=dt^2+et+f.\]
Our model says that the 
\[p_i=\frac{p_i^T-\tau_{\rm mean}(time_{i-1})}{\tau_{\rm sd}(time_{i-1})},\]
given the past, are a normally distributed process with time-series
structure.   We fit the trend parameters $\eta = (a,b,c,d,e,f)$ using
quasi-maximum likelihood. The log-quasi-likelihood function that we use is 
\begin{equation}\label{eq:thetall}
-\sum_{i=2}^n\left[\log\tau_{\rm sd}(time_{i-1})+\frac{[\hat{p}_i^T-
\tau_{\rm mean}(time_{i-1})]^2}{2\tau_{\rm sd}(time_{i-1})^2}\right],
\end{equation} 
where $n$ is the number of durations in the day.  The function in
(\ref{eq:thetall}) would be the likelihood function if the
$\hat{p}_i$ were independent rather than following a time-series model.

In step~\ref{step3}, we fit an appropriate time series model to the
$\{\hat{p}_i\}_{i\geq1}$ sequence assuming that the noise terms are
normally distributed.  With trade duration data, we fit an ARFIMA
model with orders chosen by BIC. Each time series model then says that  
the distribution of $p_i$ is the normal distribution with a fitted
mean $\hat \mu_i$ and fitted variance $\hat \sigma_i^2$ determined by
the specific model.

Instead of maximizing the quasi-log-likelihood (\ref{eq:thetall}), we
could attempt to find the joint MLE of $\eta$
and the parameters of the ARFIMA model. The log-likelihood for both sets of
parameters is not  (\ref{eq:thetall}), but rather
\begin{equation}\label{eq:thetall2}
-\sum_{i=2}^n\left[\log[ \tau_{\rm sd}(time_{i-1}) \sigma_i
  ]+\frac{[\hat{p}_i^T- 
\tau_{\rm mean}(time_{i-1}) - \tau_{sd}(time_{i-1})
\mu_i]^2}{2\tau_{\rm sd}(time_{i-1})^2 \sigma_i^2}\right],
\end{equation} 
where $\mu_i$ and $\sigma_i$ are functions of the ARFIMA parameters
that specify the mean and standard deviation of $p_i$ given the
past. Note that (\ref{eq:thetall}) is the 
special case of (\ref{eq:thetall2}) when $\mu_i=0$ and $\sigma_i=1$.
Starting with $\mu_i=0$ and $\sigma_i=1$, steps~\ref{step2}
and~\ref{step3} 
should be iteratively repeated, using (\ref{eq:thetall2}) in
step~\ref{step2}instead of (\ref{eq:thetall}), until the parameter estimates
converge. Specifically, after the first round of  steps~\ref{step2}
and~\ref{step3}, set $\mu_i$ and $\sigma_i$ respectively to the estimated
mean and standard 
deviation of $p_i$ given the past as fit by the ARFIMA model in
step~\ref{step3}.  For later iterations, use the estimated $\mu_i$ and $\hat
\sigma_i$ to refit the trend  parameters by maximizing the
log-likelihood in equation \ref{eq:thetall2}.  With the new estimated
trend parameters, refit the ARFIMA parameters and alternated until the
estimated parameters converge.  We compared this procedure to the
quasi-likelihood maximization described above and found that the
ARFIMA parameter estimates change negligibly (usually less than 1\%).
The trend parameters sometimes change as much as 15\%, but the
changes do not translate into noticeable changes in predictions.  (We
give more evidence of this last claim in
Section~\ref{subsec:prediction}.) The
empirical results that we report in Section~\ref{sec:expres} use the
quasi-likelihood maximization.

\subsection{Prediction}  \label{subsec:prediction}
When we wish to predict log-durations on a new day (which we will call {\em
  test data}), we start with the fitted
model based on the previous day's data (which we will call {\em
  training data}).  We carry forward the conditional CDF $\hat
F(\cdot|T_{i-1})$ and coefficients from the ARFIMA process that were
estimated from the training data. For the the intraday trend, we assume
that a new coefficient vector $\eta = (a,b,c,d,e,f)$ is needed each day.
We start by using the estimated $\hat
\eta=(a_0,b_0,c_0,d_0,e_0,f_0)$ from the training data.  In order to
perform updates to $\hat\eta$ in real 
time as events occur, a fast update is needed for $\hat{\eta}$.
We propose the following penalized least squares estimation (LSE) method.
Whenever a new pair of $(time_{j-1}, p^T_j)$
arrives from the test data, we update our estimated $\hat 
\eta$ as follows.  Choose $(\hat a, \hat b, \hat c)$ to minimize
\[\sum_{i=1}^j (p^T_i - a \times time_{i-1}^2 - b \times time_{i-1} -c )^2 +
\lambda(a-a_0)^2 +\lambda(b-b_0)^2 + \lambda(c-c_0)^2,\]
and  set 
\begin{equation}\label{eq:tauhat1}
\hat{\tau}_{\rm mean}(time_{i-1})=\hat a\times
time_{i-1}^2+\hat b \times time_{i-1}+\hat c.
\end{equation}
Then choose $(\hat d,
\hat e, \hat f)$ to minimize 
\[\sum_{i=1}^j [ (p^T_i-\hat \tau_{\rm mean}(time_{i-1}))^2 - d
\times time_{i-1}^2 - e \times time_{i-1} -f ]^2 + \lambda(d-d_0)^2 +
\lambda(e-e_0)^2 + \lambda(f-f_0)^2,\]
and  set 
\begin{equation}\label{eq:tauhat2}
\hat{\tau}_{\rm sd}(time_{i-1})=\sqrt{\hat d\times
time_{i-1}^2+\hat e \times time_{i-1}+\hat f.}
\end{equation} 
Empirical results depend little on the value of $\lambda$ for
$\lambda$ in a large interval. We use $\lambda= 10$ in our calculations.

A more time-consuming, but perhaps more principled, method of updating
the trend parameters would be to compute the posterior mode (PM) after
each trade event.  We could use the negatives of the penalizations in
the LSE method log-priors and maximize the sum of those log-priors and
the log-likelihood (\ref{eq:thetall2}).
The quality of the LSE approximation, compared to PM is illustrated in 
Figures~\ref{fig:PMvsLSE} and~\ref{fig:individualLL_LSEvsPM}.  
We see that the two update methods produce intraday trends that
are very similar with predictions of comparable quality.  Because LSE
works many times faster 
than PM, we use LSE for prediction in the remainder of the paper.
\begin{figure}[htb]
\begin{center}
\includegraphics[width = 1.0\textwidth, height = 3.0in]{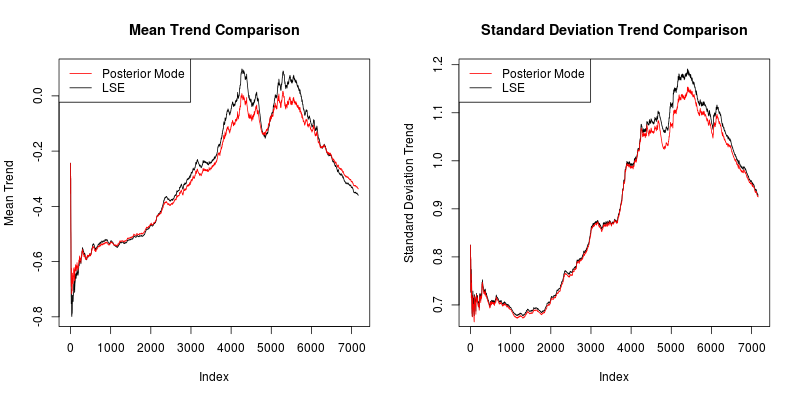}
\caption{Estimated intraday trends in both mean and standard
  deviation for one stock on one day. The black line gives the
  estimates based on LSE, and the red
  line gives the estimates based on PM.}\label{fig:PMvsLSE}
\end{center}
\end{figure}

\begin{figure}[htb]
\begin{center}
\includegraphics[width = 0.7\textwidth, height = 3.0in]{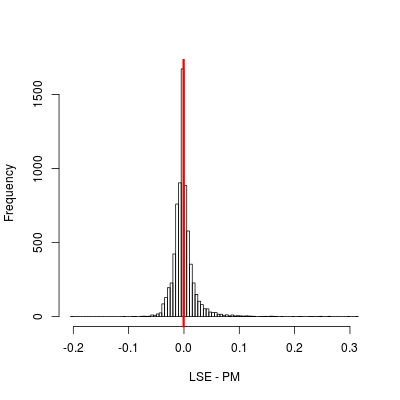}
\caption{Histogram of differences of individual prediction
  log-likelihoods (one for each predicted duration) between LSE and PM for JPM  
on one day. The mean and median are both close to 0, and the average
of the individual prediction log-likelihoods is $-2.15$, so the
percentage differences are also small.  The other stocks and days have
similar patterns. }\label{fig:individualLL_LSEvsPM} 
\end{center}
\end{figure}

In Section~\ref{sec:expres}, we evaluate our model
fit and compare it to the fits of other models.  We fit all models
using training data and then base the evaluations and
comparisons on test data. We use the prediction
log-likelihood and goodness-of-fit tests based on generalized residuals.

\subsection{Predictive Distribution and Log-Likelihood}\label{sec:preddist}
In this section, we show how to use estimates from the training data
along with the continuously updated intraday trend described above to
compute the predictive distribution of test data.

The fitted conditional CDF of the next $T_i$ in the test data given
the past can be constructed using 
(\ref{eq:tcdf}).  The time series model says that our estimate of the
conditional CDF $G_i$ is $\hat
G_i(u)=\Phi([u-\hat\mu_i]/\hat\sigma_i)$, where $\hat\mu_i$ and
$\hat\sigma_i$ are based on the estimated ARFIMA parameters from the
training data along with the past durations in the test data.
Our estimate of $H_i$ from the training data is $\hat H_i=\hat
F_i^{-1}(\Phi)$, so the estimated conditional CDF of $T_i$ is
\begin{equation}\label{eq:estcdf}
\hat{F}_i^*(t)=\Phi\left(\frac{1}{\hat\sigma_i}\left[
\frac{\Phi^{-1}(\hat F_i(t))-\hat\tau_{\rm mean}(time_{i-1})}
{\hat\tau_{\rm sd}(time_{i-1})}-\hat{\mu}_i\right]\right),
\end{equation} 
where $\hat{\tau}_{\rm mean}$ and $\hat{\tau}_{\rm sd}$ come from
(\ref{eq:tauhat1}) and (\ref{eq:tauhat2}) respectively, and are
recomputed each time that a new event occurs.

In order to compute the prediction log-likelihood of the test data, we
need the density corresponding to the CDF in (\ref{eq:estcdf}).  This
is obtained by standard calculus operations as
\begin{equation}\label{eq:estdens}
\hat{f}_i^*(t)=\frac{1}{\hat{\tau}_{\rm sd}(time_{i-1})}
\phi_{\hat\mu_i,\hat\sigma_i}\left(\frac{\Phi^{-1}(\hat F_i(t))
-\hat\tau_{\rm mean}(time_{i-1})}{\hat\tau_{\rm sd}(time_{i-1})}\right)
\times\frac{\hat{f}_i(t)}{\phi_{0,1}(\Phi^{-1}(\hat{F}_i(t)))},
\end{equation} 
where, $\phi_{m,s}(\cdot)$ denotes the density of the normal distribution
with mean $m$ and standard deviation $s$, and $\hat{f}_i$ is the
density that corresponds to $\hat{F}_i$.  
The predictionlog-likelihood for the test data $\{T_i\}_{i=1}^n$ is 
$\sum_{i=1}^n\log (\hat f_i^*(T_i))$.

The form (\ref{eq:estdens})
has a convenient interpretation for comparing our model to some
submodels.  For example, if we ignore the intraday trend, then we just
set $\hat{\tau}_{\rm mean}(\cdot)\equiv0$ and $\hat{\tau}_{\rm sd}(\cdot)\equiv1$.
If we wish to ignore the ARFIMA modeling, we just set $\hat\mu_i=0$
and $\hat\sigma_i=1$.

The {\em final generalized residual} corresponding to $T_i$ from the
test data is computed by substituting $T_i$ for $t$ in (\ref{eq:estcdf}):
\begin{equation}\label{eq:fgr}
c^*_i = \hat F_i^*(T_i) =\Phi\left(\frac{1}{\hat\sigma_i}\left[
\frac{\Phi^{-1}(\hat F_i(T_i))-\hat\tau_{\rm mean}(time_{i-1})}
{\hat\tau_{\rm sd}(time_{i-1})}-\hat{\mu}_i\right]\right),
\end{equation} 
If the model fits well, then the $c^*_i$ should look like a sample of
independent uniform random variables on the interval $(0,1)$.
We will present goodness-of-fit tests 
results based on $\{c^*_i\}_{i=1}^n$ in Section~\ref{sec:expres}. 

\subsection{Additional Explanatory Variables} \label{subsec:bpi}

In a marked point process, marks are observed along with the target
events. These marks, or additional explanatory variables, may have an
impact on the distribution of log-duration. In the ACD model and its
variants, the 
marks' information is incorporated in the autoregressive formula
(\ref{eqn:acd2}). Analogously, in our model, 
it is natural to incorporate the additional variables' effects in the
parametric time-series model (step~\ref{step3}).  

A straightforward way to incorporate additional variables into an
ARFIMA model is to extend the
ARMA model with regression  \cite{book}.  We call the extension {\em ARFIMA 
with regression}.
Various methods for estimating ARFIMA models have been reviewed in
  \cite{ARFIMAsurvey}. 
In the trade duration example in Section~\ref{sec:expres}, we use ARFIMA with
regression in order to incorporate an additional explanatory variable
into our model.


\section{Experimental Results}\label{sec:expres}
In this section, the HSDM framework is applied to the trade flows from
the limit order books of four stocks (IBM, JC Penny, JP Morgan, and Exxon Mobil)
on the New York Stock Exchange (NYSE) for selected dates between
6 July 2010 and 29 July 2010, 18 consecutive trading days.  In
Section~\ref{sec:modelbuild}, the first eight days are used to build
the model and discover patterns. In Section~\ref{subsec:diagnostics},
we use the remaining days to validate the model and to make
comparisons between HSDM and benchmark models. Each day is used as
training data to estimate parameters and the following day is used as
test data for prediction and goodness-of-fit tests.  This pattern of
training data followed by test data is used during both model building
and validation.

Since the focus of this paper is on the dynamics of trade flow, the
limit order book data are preprocessed, and a list of triples 
$(time_i, T_i, BPI_i)$ is produced, where $time_i$ is the clock time,  
$T_i$ is log-duration, and $BPI_i$ (book pressure
imbalance) is our additional explanatory variable.
Table~\ref{table:triple} gives a sample of three consecutive such triples.
\begin{table}[hbt]
\centering
\caption{A sample of trade flow data.  Clock time is in milliseconds
  since midnight, and log-duration is in log-milliseconds since
  previous trade.}\label{table:triple} 
\begin{tabular}{l ccc } 
\hline
clock time & log-duration & book pressure imbalance  \\
\hline 
43026177 & 3.91 & -0.588 \\
43026179 & 6.82 & -0.134 \\
43026180 & 6.10 & -1.946  \\
\hline
\end{tabular}
\end{table}
We observe such a triple whenever a trade occurs. 
 Book pressure imbalance is defined as the 
  logarithm of the ratio of the number of sell-side shares at the ask
  price to number of buy-side shares at the bid price. In general, it
  measures the imbalance of demand between the buy and sell sides of
  the market. BPI varies from time to time and actually changes much 
more frequently than trades occur. In this paper, we keep track of
the book pressure imbalance only when a trade happens. Moreover, since
the focus in this article is the duration between trades, it is the
degree of imbalance rather than the direction of imbalance that
matters. Thus, we use the absolute value of book pressure imbalance,
$|BPI|$, as our additional explanatory variable.

In Section~\ref{sec:modelbuild}, we present empirical evidence that
motivates each of the stages in the hierarchy of the HSDM framework.
Section~\ref{sec:bench} describes the benchmark models to which we
compare HSDM in Section~\ref{subsec:diagnostics}.

\subsection{Model Building}\label{sec:modelbuild}
Here we show why we chose the particular stages in the hierarchy of
the HSDM model of Section~\ref{sec:hsdm}.  The choices are based on a
period of model building data ranging from 6 July 2010 to 16 July
2010.  It is widely accepted that 
duration processes in market microstructure data are self-exciting
(especially intertrade durations). Figure~\ref{fig:selfexciting}
\begin{figure}[hbt]
\begin{center}
\includegraphics[width = 0.8\textwidth, height =
4in]{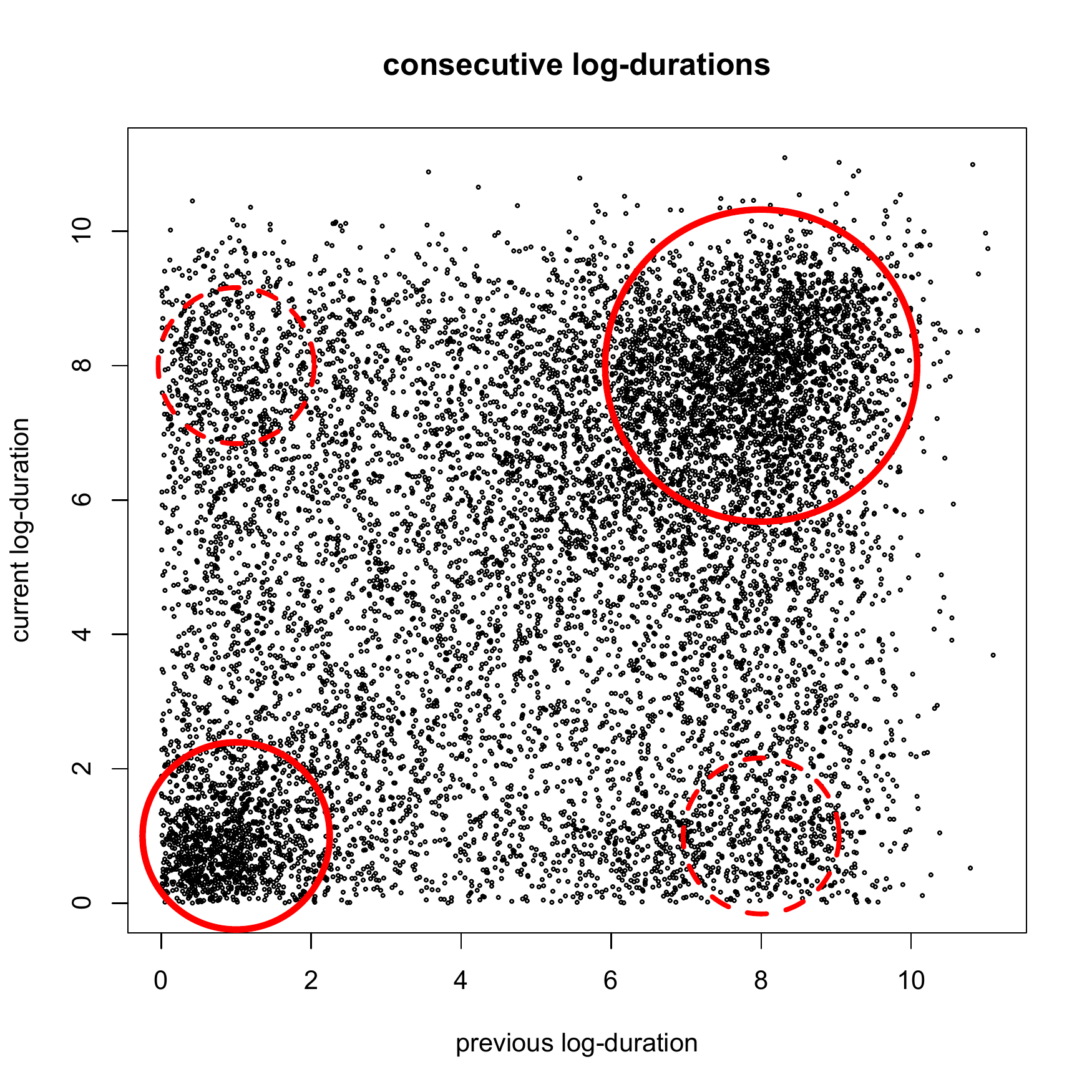} 
\caption{Current log-duration (vertical axis) vs. previous
  log-duration (horizontal axis) for one stock on one day. Two key
  features are noted by the circled regions.  The solid circles indicate
  the self-exciting nature by showing the clustering of pairs with
  both durations being short and with both durations being long. The
  dashed circles indicate the bimodal nature by showing secondary
  modes at both long and short durations regardless of the length of
  the other duration.}\label{fig:selfexciting} 
\end{center}
\end{figure}
shows a scatter plot of the pairs of consecutive
log-durations for one stock on one day.  The plot shows that (i)  long
durations tend to be 
followed by long durations, while short durations tend to be followed
by short duration (the self-exciting property) and (ii) the marginal
distribution of log-duration and the conditional distribution of
log-duration given the previous log-duration are bimodal. Since we
expect the most recent durations to carry the most information,
step~\ref{step1} of the 
estimation procedure employs nonparametric kernel
conditional density estimation as described by Hall, Racine and Li
(2004) \cite{ABC7}, conditioning on the previous log-duration
$T_{i-1}$. A technical issue arises due to the discreteness
  of durations, as they are measured to the nearest millisecond.  We explain
  this issue in more detail in Appendix~\ref{sec:noise} along with
  how we deal  with
 it.  In particular,  we explain why it makes sense
  to base the estimation on the logarithms of the durations.
The estimated conditional CDF of $T_i$ given $T_{i-1}$ is denoted by
$\hat F_i(\cdot)$. Figure~\ref{fig:npcdens}
\begin{figure}[htb]
\begin{center}
\includegraphics[width = 1.0\textwidth, height = 5in]{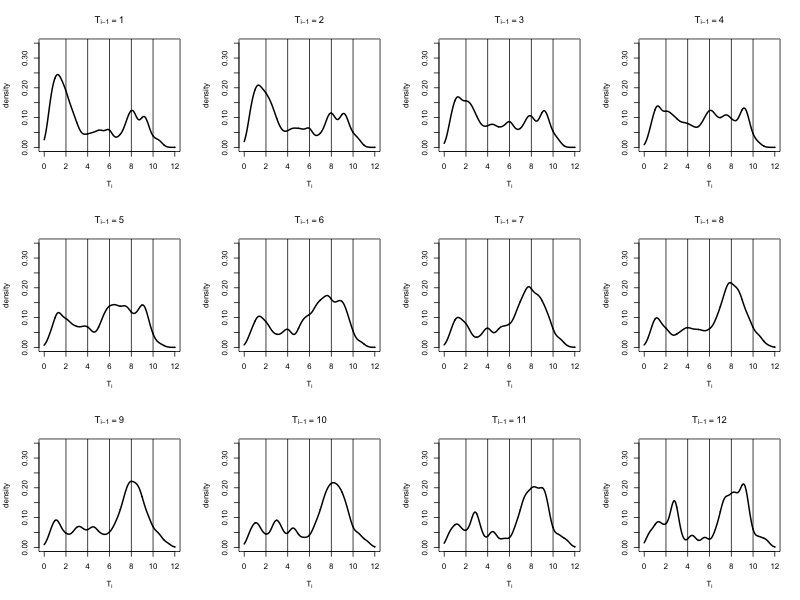}
\caption{Conditional density of $T_i$ given $T_{i-1}$ for 12 different
values of $T_{i-1}\in\{1,\ldots,12\}$.  These range from
$2.7\times10^{-3}$ seconds to 2.7 minutes.}\label{fig:npcdens}
\end{center}
\end{figure}
shows the conditional densities
 that were estimated conditional on
different values of 
the previous log-duration
 $T_{i-1}$ for one day of one stock. The
conditional densities capture the self-exciting feature of trades.
Note how the conditional density is highest near zero when the
previous duration is short, but when the previous duration
 is long, the 
density is highest at larger values. No matter what the previous
duration $T_{i-1}$ is, the 
distribution of the current duration $T_i$ has two local modes. And as
the previous 
duration increases, both the height and the location of the second
mode increase. The bimodal characteristic partly explains why
parametric conditional duration models with a unimodal residual
distribution cannot capture the dynamics of trade flow very well. 
Nonparametric conditional density estimation in step~\ref{step1}
captures not only some short memory information but also the bimodal
nature of the conditional distribution.

In step~\ref{step2} of the estimation, we introduce intraday trends
for both the mean and variance of the $\{p_i^T\}_{i\geq1}$ sequence.
Figure~\ref{fig:intradayTrend} 
\begin{figure}[htb]
\begin{center}
\includegraphics[width = 1.0\textwidth, height = 3in]{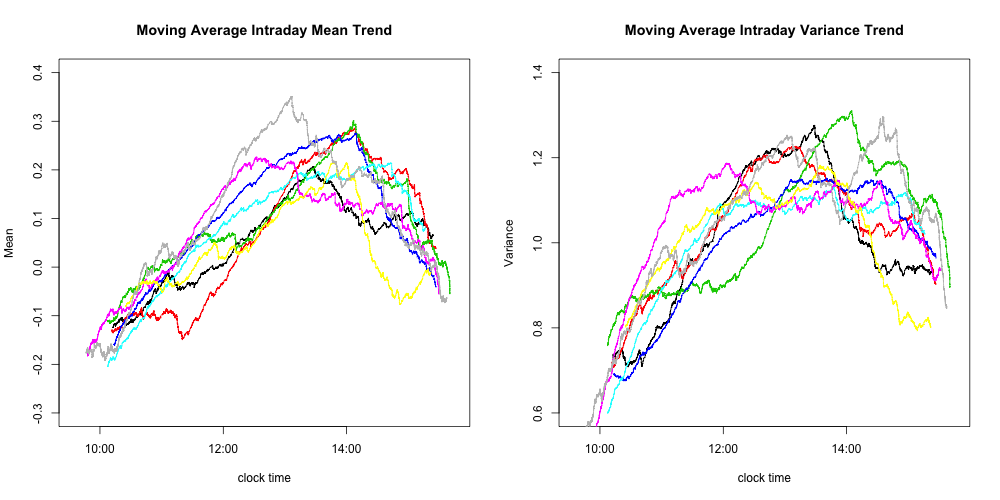}
\caption{Fitted intraday trends on eight different days (one color for each
  day) for both the mean and variance of $p_i^T$ for a single stock (JPM).
  The trends were fitted as moving average trends with window length equal 2000 trades fit to the $\hat{p}_i^T$.}
\label{fig:intradayTrend}
\end{center}
\end{figure}
shows empirical evidence for those trends for eight different days and one
stock.  The mean 
trends in Figure~\ref{fig:intradayTrend} are computed as moving
averages $m_i$ of the $\hat{p}^T_i$ sequences throughout each day.  The
trends for variance were computed as moving averages of
$(\hat{p}^T_i-m_i)^2$ throughout the day.
The shapes suggest that an intraday trend is present in the data and
that a quadratic shape might provide a good fit. The mean and the
variance of $p_i^T$ have similar intraday patterns, but they change
slightly from day to day.  This apparent change motivates our online
estimation of the parameters described in Section~\ref{subsec:prediction}.

After detrending the $\hat{p}_i^T$ in step~\ref{step2}, we compute their
autocorrelation function and partial autocorrelation function, which appear
in Figure~\ref{fig:acf} for a typical trading day of JPM. These plots
are typical of the patterns that we see across all of the four stocks
and all days. 
\begin{figure}[htb]
\begin{center}
\includegraphics[width = 1.0\textwidth, height = 2.5in]{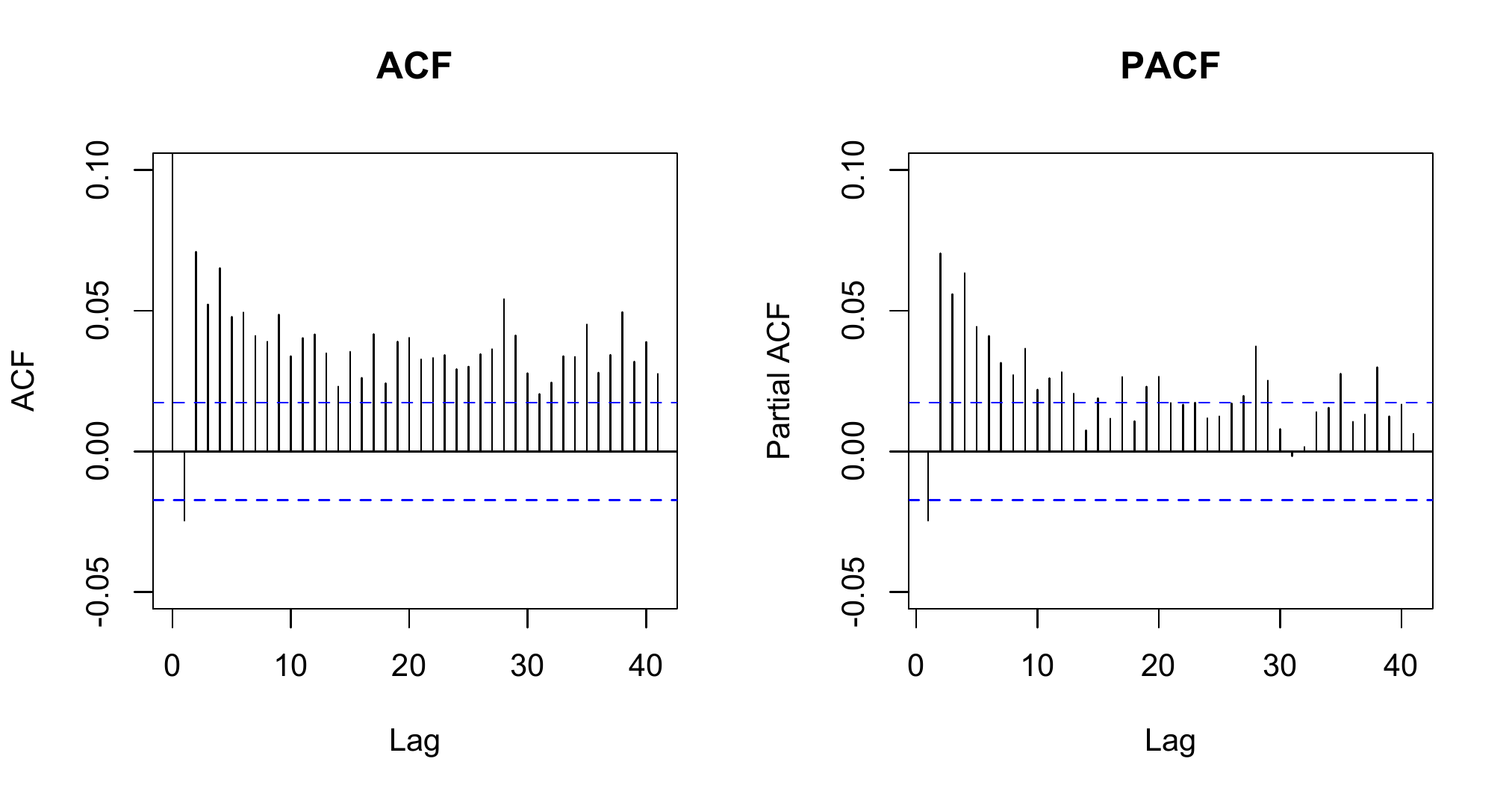}
\caption{Sample ACF and PACF of transformed generalized residuals
  after detrending.  The lag-zero autocorrelation of 1 is truncated in
the ACF plot.} 
\label{fig:acf}
\end{center}
\end{figure}
The patterns in these plots 
suggest the presence of long memory. The
negative lag-one autocorrelation suggests that the conditional density
estimation may be overfitting the lag-one dependence.

In step~\ref{step3} of the estimation, we begin with a long-memory
time series model having no exogenous variables.
Appendix~\ref{sec:arfima} provides some detail about both 
the time series model that we use (ARFIMA) and how we incorporate
regression into that model.  The particular model that we choose is
ARFIMA(0,$d$,1), which has the form
\begin{equation}
\label{eqn:ARFIMA01}
(1 - B)^d P_i  = (1 - \theta B)Z_i, \ \ where \ \ Z_i \sim N(0, \sigma^2).
\end{equation}
With our data sets, the differencing parameter $d$ is typically
estimated to be around 
0.1 (with standard error around $0.01$), suggesting significant long memory.
We chose this from the family of  ARFIMA($p,d,q$) models by 
minimizing the BIC score among the potential choices of
$p\in\{0,1,2,3\}$ and $q\in\{0,1,2,3\}$.  

\subsection{The Benchmark Models}\label{sec:bench}
In this section, we implement four common members of the ACD family of
models as benchmarks.  These include exponential ACD, semiparametric
ACD, exponential FIACD, and semiparametric FIACD.  
For each of these models, intraday patterns are estimated directly from 
the durations by fitting a cubic spline with knots at each full hour
of clock-time.  After fitting the spline, the $i$th duration is divided by
the fitted spline value at $time_{i-1}$. The resulting ratios are used
as the input data for the ACD model and its variants as Jaisk did in
\cite{FIACD}. For  
predicting test data, we use the estimated intraday pattern based on the
previous day's (training) data. 
All of the benchmark models start with the formula 
\[X_i = \Psi_i \epsilon_i,\]
where $X_i$ is a duration (divided by the intraday trend), $\Psi_i$ is
the conditional mean of $X_i$, and $\epsilon_i$ are IID from a long-tailed
distribution with mean 1. The exponential ACD and exponential FIACD,
assume that $\epsilon_i$ has the standard exponential distribution
while the
semiparametric versions allow $\epsilon_i$ to have a general density
$f(\cdot)$ that is fit by kernel density estimation. For the ACD and exponential
ACD models, the conditional mean of $X_i$ evolves as
\[\Psi_i = \omega + \alpha X_{i-1} + \beta \Psi_{i-1}.\]
For the fractionally integrated versions, the conditional mean evolves
as
\[ (1-\beta) \Psi_i = w + [1-\beta - (1-\alpha-\beta)(1-B)^d] X_{i-1}.\]
After fitting these models we compare them all to HSDM 
in terms of prediction log-likelihood and various model diagnostics as
described in Section~\ref{subsec:diagnostics}.

\subsection{Model Comparisons} \label{subsec:diagnostics}
In this section, we compare the five models based
on  prediction log-likelihood and a number of diagnostic tests using
the validation data from 19 July 2010 through 29 July 2010.
Diebold, Gunther and Tay (1998) \cite{denFORCAST} (henceforth DGT)
proposed a method of evaluating prediction models based on the 
probability integral transform.   If the predictive
distribution of the $i$th observation $X_i$ given the past has the CDF
$F_i(\cdot)$, then the sequence of values $\{F_i(X_i)\}_{i\geq1}$ forms an
IID sample of uniform random variables on the interval (0,1). 
Of course, we don't know $F_i$, but each model provides a fitted
$\hat{F}_i$ for each $i$.  We can then see to what extent the sequence
of final generalized residuals,
$\{\hat{F}_i(X_i)\}_{i\geq1}$, looks like a sample of IID uniform
random variables on the interval (0,1). 
The empirical distribution of the sequence should look like a
uniform distribution, and the sequence should not exhibit any autocorrelation.

For the HSDM model, the final generalized residuals, $c_i^*$ come from
(\ref{eq:fgr}). Each of the other models has a
corresponding final generalized residual to be tested, i.e. $c_i^* =
\hat F(X_i/\hat{\Psi}_i)$, where $\hat F(\cdot)$ is the estimated CDF of
$\epsilon$ for each model. In the exponential ACD and exponential
FIACD models, $\hat F(\cdot)$
is the CDF of the exponential distribution with mean 1, while in the
semiparametric ACD and semiparametric FIACD,  $\hat F(\cdot)$ is an
estimate based on kernel density estimation using the fitted
$\epsilon_i$ values as suggested by \cite{FIACD}.  In
Appendix~\ref{app:tail} we give more detail on the kernel density
estimation. In particular, we explain why it makes sense to base the
estimation on the $\log(\epsilon_i)$ values.

Bauwens, Giot, Grammig and Veredas (2000)
\cite{COMPARE} compared some of the most popular conditional duration models,
including ACD, log ACD, threshold ACD, SCD and SVD,  by means of the DGT
method. Although they found that these models generally work well on the
price duration process and the volume duration process, none of them work
on the trade duration process.  

There is one important distinction between the 
diagnostics that we compute and those computed in most
other papers on self-exciting point process models.  As in other
papers, we first fit models to training data.  The difference is that
we compute the final generalized residuals by using the fitted
models to predict test data. Most papers compute their diagnostics
from the final generalized residuals obtained by predicting the same
training data that were used to fit the 
models.  There are two main reasons for using test data to perform the
diagnostics rather than using training data.  First, it is well-known
that virtually all 
statistical models fit better to the data from which they were
estimated than to new data that were not used for their estimation.
It is good statistical practice to evaluate the fit of every model on
test data, if such data are available.  Second, we are comparing a
number of models that are semiparametric along with some that do not
form a nested 
sequence.  Traditional likelihood-ratio tests are useful for comparing
nested parametric models in order to see whether the additional
parameters provide significant improvement or merely overfit.  With
semiparametric models, the theory of likelihood-ratio tests is still
being developed.  In order to minimize the chance of overfitting with
semiparametric models, it is good practice to evaluate them with test
data that were not used in the fitting.  (This procedure is good
practice even with parametric models.)  If a model overfits the
training data, it will make noisy predictions with test data.  In
comparing two or more models, comparing their predictions based on
test data is the safest way to avoid choosing a model that was overfit.

\subsubsection{Prediction Log-Likelihood}
In this section, we compare the five models based on prediction
log-likelihood for test data.  This is essentially a comparison based
on how high is each model's predictive density at the observed test data.
The larger the prediction log-likelihood, the better the prediction
is. Figure~\ref{fig:ACDvsHSDM}
\begin{figure}[htb]
\begin{center}
\includegraphics[width = 1.0\textwidth, height = 3.5in]{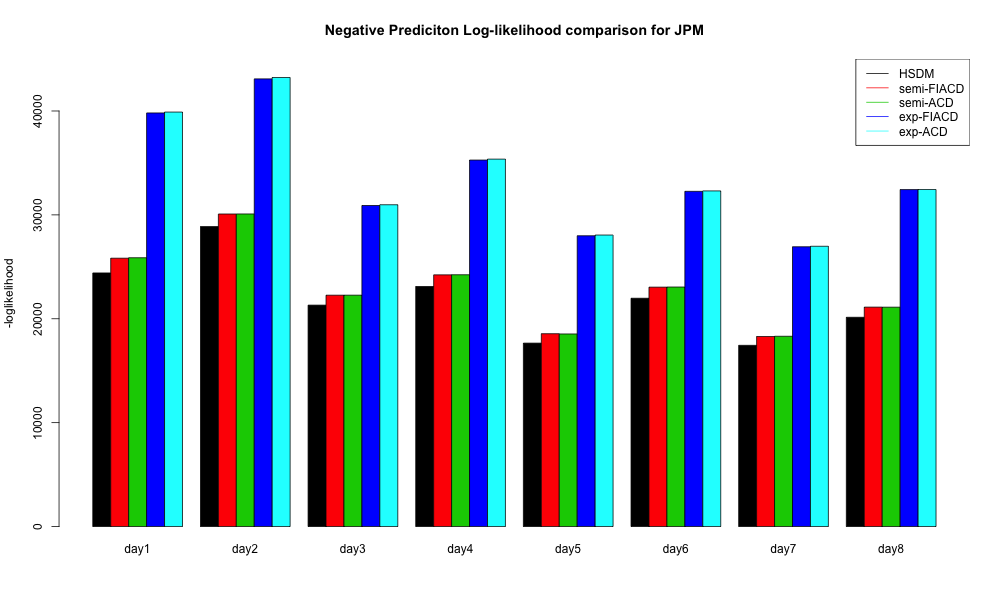}
\caption{Comparison between HSDM model and benchmark models in terms
 of negative prediction log-likelihood. Smaller is better.} 
\label{fig:ACDvsHSDM}
\end{center}
\end{figure}
shows the negative prediction
log-likelihood for JPM on the eight consecutive test days in the
validation data for all five models. The HSDM model (black bars) is
consistently better than all of the benchmark models. And within the
four benchmark models, semi-parametric models are better than
exponential models. ACD models and FIACD models exhibit similar
performance. The other three stocks (IBM, XOM, JCP) have similar patterns. 
 the HSDM and benchmark model perform on each individual
observation is of great interest.  
Figure~\ref{fig:HistIndividualLL} 
\begin{figure}[htb]
\begin{center}
\includegraphics[width = 0.6\textwidth, height = 2.8in]{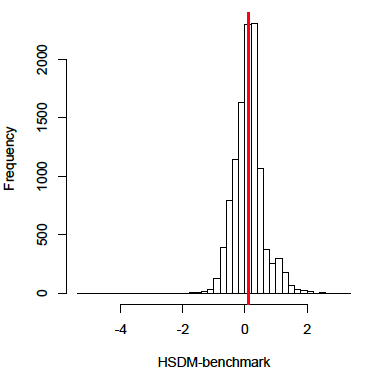}
\caption{Histogram of differences between individual log-likelihoods,
  HSDM minus the best benchmark model (semiparametric-FIACD) on a
  typical day of JPM. The red vertical line is the mean ($\approx
  0.122$). The median is 0.124, and the percentage of differences that
  are positive is 62.4\%.}  
\label{fig:HistIndividualLL}
\end{center}
\end{figure}
compares the performance of HSDM and the best benchmark model,
semiparametric FIACD, on a typical day of 
JPM. The plot 
displays the histogram of differences of individual prediction
log-likelihood between HSDM and semiparametric FIACD. The positive
part of the histogram
corresponds to those observations on which HSDM has a better
prediction than semiparametric FIACD, while the negative part corresponds to
observations on which HSDM performs worse than semiparametric FIACD. The mean of
the difference is around 0.122 (denoted by the red vertical
line), the median is around 0.124, and the percentage of positive
differences is 62.4\%. 
All of the other days and stocks have
similar patterns.  

\subsubsection{Uniform Test}\label{sec:utest}
In this section, we compare each set of final generalized residuals to
the uniform distribution on the interval (0,1) by means of the
Kolmogorov-Smirnov (KS) test. 
The left subfigure in Figure~\ref{fig:qqplot}
\begin{figure}[hbt]
\begin{center}
\includegraphics[width = 1.0\textwidth, height =
2.8in]{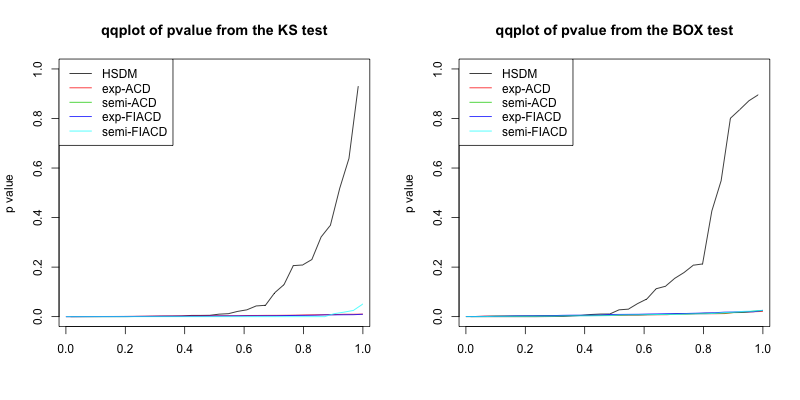} 
\caption{Q-Q plots of $p$-values from KS test and Ljung-Box test on HSDM and
  benchmark models.}\label{fig:qqplot}
\end{center}
\end{figure}
shows the Q-Q plot of 32
$p$-values from KS tests for each of the five models (8 test days for
each of 4 stocks). If the final
generalized residuals were really sampled from 
their predictive distributions,
then the $p$-values should be uniformly
distributed on the interval $(0,1)$ and the the Q-Q plot should be
around the 45-degree line. If the final generalized residuals come
from different distributions, the $p$-values should be stochastically
smaller than uniform on $(0,1)$.

Since each model estimates a predicitive distribution from the observed
data (including training data), we cannot expect the $p$-values to be
uniformly distributed.  When the estimated distributions come from a
finite-dimensional parametric family, there are modifications
available to the KS test so that the $p$-values have uniform
distribution asymptotically.  When the predictive distributions are estimated
nonparametrically or semiparameterically, the appropriate
modifications have not yet been determined.  Nevertheless, the KS test
statistics (or equivalently their $p$-values) still give a means for
comparing models based on how close to uniform the final generalized
residuals appear to be.  From the comparison in
Figure~\ref{fig:qqplot}, it is clear the benchmark models have
(empirically) stochastically smaller $p$-values than HSDM.
The HSDM $p$-values are still stochastically smaller than the
uniform distribution, but they are much larger than those for the
benchmark models.

The reason that the $p$-values from the exponential models are all so small
is that the data come from a distribution with a much fatter tail than
that of the 
exponential distribution. As a matter of fact, no popular parametric model
can perform satisfactorily because the empirical log-duration has a
bimodal shape.  The semiparametric models perform relatively better
but still have very small $p$-values. 

\subsubsection{Autocorrelation Test}\label{sec:actest}
In this section, we assess the degree of autocorrelation in the final
generalized residuals. We used the Ljung-Box test with lags of 5, 10,
and 15. The test is conducted on each of the 32 pairs of stock/test
day for each model.  
The right subfigure in Figure~\ref{fig:qqplot} shows Q-Q plots of 32
$p$-values for the Ljung-Box test with lag 10 for all five models. The results
of lags 5 and 15 are similar. If there were no autocorrelations, the
$p$-vaules should be uniformly distributed on the interval $(0,1)$.
If there are autocorrelations, the $p$-values should be stochastically smaller.
The four benchmark models have $p$-values that are stochastically much
smaller than those of HSDM.
of the Box test. Although the qqplot of HSDM shows that the p values
are not from a standard uniform distribution, it is significantly
better than all of the four benchmark models, which almost have all 0
$p$-values.  
The reason that the $p-$values are so low for the benchmark models is
that they don't capture the 
information of the most recent duration very well. The HSDM model captures
this information non parametrically, which helps eliminate lag-one
autocorrelation from the final generalized residuals. In summary, the HSDM
model captures the time dependency significantly better than the
benchmark models.

\subsection{Additional Explanatory Variables} \label{subsec:AEV}
As mentioned earlier, the effects of
exogenous variables can be incorporated in step~\ref{step3} of the HSDM
framework by extending the parametric time series model with
regression. As an example, we took into account book pressure
imbalance (BPI) (defined in 
  Section~\ref{sec:expres}) as additional explanatory variables. Other
  variables, such as spread and volume, may be considered
  similarly. Figure~\ref{fig:BPIeffect} 
\begin{figure}[htb]
\begin{center}
\includegraphics[width = 3.5in, height = 3.0in]{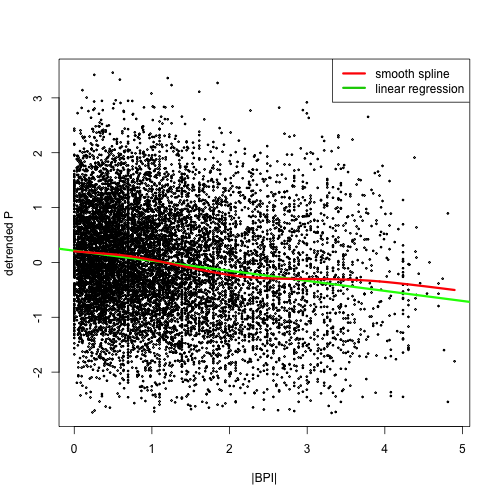}
\caption{Effects of BPI on detrended P. Detrended P vs absolute value of BPI on a typical day of JPM. A significant downward trend is noticeable. }
\label{fig:BPIeffect}
\end{center}
\end{figure}
plots $\hat{p}_i$ (detrended
  $\hat{p}^T_i$) against $|{\rm BPI}_i|$ on a typical day of JPM. The red curve
  is a fitted smoothing spline, while the green line is the fitted
  linear regression.  
Apparently, as the absolute value of BPI increases, which means that
there is more imbalance between buy side and sell side, $\hat{p}_i$
becomes smaller, leading to a shorter duration as expected. And the
effect is close to a linear relationship.
We choose how many lags of BPI to include in the
regression by BIC. It turns out that the number of lags varies by
stock. For example, the inclusion of two lags of BPI improves the 
prediction on JPM and XOM significantly, while there is
no noticeable improvement with IBM and JCP. There is no reason that
the same exogenous variables should be included in the models for all
stocks, hence each stock will be modeled either with or without
inclusion of BPI as we determine during the model building stage.
Equation (\ref{eqn:arfimaR}) is the ARFIMA with regression model for
those stocks that use two lags of BPI.
\begin{equation}
\label{eqn:arfimaR}
(1-B)^d ( P_i - b_0 - b_1 |BPI_{i-1}| - b_2 |BPI_{i-2}| ) = (1 -
\theta B) Z_i,  \ \ Z_i \sim N(0, \sigma^2) 
\end{equation}
Figure~\ref{fig:llimprovementBPI} 
\begin{figure}[htb]
\begin{center}
\includegraphics[width = 4.5in, height = 2.5in]{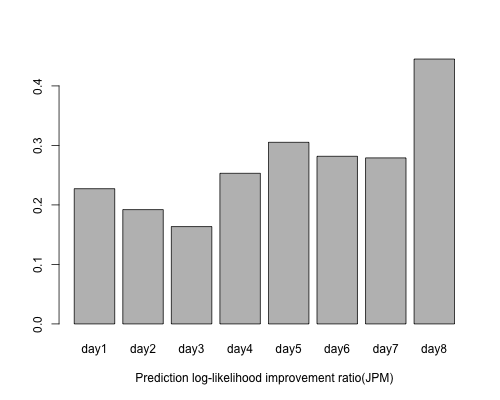}
\caption{Increase in prediction log-likelihood by incorporation
  of $|BPI|$ for JPM as a fraction of the increase of HSDM (without
  BPI) over semiparametric-FIACD, i.e., equation (\ref{eq:pllratio}).}
\label{fig:llimprovementBPI}
\end{center}
\end{figure}
shows the increase in prediction
log-likelihood from incorporating BPI as in \ref{eqn:arfimaR} as a
fraction of the amount by which HSDM (without $|BPI|$) improves over
the best benchmark 
model (semiparametric-FIACD) for each of 8 days for JPM.  That is, the
plot shows 
\begin{equation}\label{eq:pllratio}
\frac{ll_{\rm HSDM\ with\ BPI} - ll_{\rm HSDM\ without\ BPI}}{ll_{\rm HSDM\ without\
    BPI} - ll_{\rm semiparametric-FIACD}}.
\end{equation} 
The consistent large positive ratios across 8
validation days' data indicate that the incorporation of BPI as in
equation \ref{eqn:arfimaR} further improves the model significantly on
JPM.  

The model constructed in this section is merely an illustration of how
one might incorporate an exogenous variable into HSDM,
hence, we did not include BPI in the benchmark models for comparison.


\section{Discussion and Conclusion}\label{sec:disc}

We proposed a semiparametric framework for estimating the joint
distribution of a marked point process. In particular, we applied
our framework to trade duration processes. Using validation data that were
not used to fit the models, the DGT evaluation methods (Diebold,
Gunther and Tay, 1998 \cite{denFORCAST}), 
show that our model does consistently better than a number of benchmark models
that are variants of the widely-used ACD model.  The evaluation methods include
Kolmogorov-Smirnov tests for uniformity of final generalized residuals
and Ljung-Box tests for lack of autocorrelation.
Bauwens, Giot, Grammig and Veredas (2000)
\cite{COMPARE} claimed that the parametric ACD model and all of its
parametric variants 
fail to pass both the Kolmogorov-Smirnov test and the Ljung-Box
test. This paper also shows that even semiparametric ACD models and
their variants  perform poorly in the sense of DGT evaluation
methods, while HSDM  shows a consistent
improvement over the benchmark models. In addition, the HSDM model has a
consistently better performance than the ACD model and its variants in
the sense of 
prediction log-likelihood on validation data. Therefore, our framework has
great potential 
for modeling the distributions of duration processes, especially
the trade duration process.

The framework has two important features. First, it recognizes that
both the shapes of distributions and the 
time dependency must be modeled. Nonparametric conditional density estimation
captures the shapes of distributions along with the most recent time
dependence. Parametric time series models capture the longer-term time
dependency.  Nonparametric estimation of the most recent time
dependency gives the model greater flexibility.  Second, our
estimation procedure adaptively fits the
intraday trend so as to capture changes that occur from day to day. 

The two features described above help to explain why the HSDM
procedure outperforms the existing ACD family on trade duration
processes.  Some of these features could be incorporated into ACD
models and their fitting.  Such incorporation will be the focus of
future work. However, every model that is based on equation
(\ref{eqn:acd1}) will continue to suffer from some of the limitations
mentioned in Section~\ref{sec:limit}.

\appendix

\begin{center}
{\bf Appendix}
\end{center}

\section{ARFIMA and ARFIMA With Regression
  Models}\label{sec:arfima} 
Equation~(\ref{eqn:arfimaApx}) shows the formula for the general
ARFIMA($p,d,q$) model for a response variable $Y$. 

\begin{equation}
\label{eqn:arfimaApx}
\left(1-\sum_{j=1}^p\phi_jB^j\right)(1-B)^d Y_i = 
\left(1 + \sum_{j=1}^q\theta_i B^j\right) Z_i, \ \ Z_i \sim N(0, \sigma^2),
\end{equation}
where $(1-B)^d$ is defined by the generalized binomial series
expansion as follows: 
\begin{eqnarray*}
(1-B)^d &=& \sum_{k=0}^{\infty}  \left( \begin{array}{c} d \\ k \end{array} \right)   (-B)^k \\
	     &=& \sum_{k=0} \frac{\Pi_{a=0}^{k-1}(d-a)(-B)^k}{k!} \\
	     &=& 1 - dB + \frac{d(d-1)}{2!} B^2 - ...,
\end{eqnarray*}
where $B$ is the backshift operator.
The {\em R} package {\tt fracdiff} \cite{fracdiff} can be
used to fit ARFIMA models. 

Suppose that we have $k$ auxiliary variables $U_1,\ldots,U_k$ that we
contemplate using to help predict $Y$.  
The ARFIMA with regression model replaces $Y_i$ in (\ref{eqn:arfimaApx})
with 
\[Y_i- \alpha_0 - \sum_{j=1}^k \alpha_j u_{i,j},\]
where $u_{i,j}$ is the value of $U_j$ that corresponds to $Y_i$, and
$\alpha_0,\ldots,\alpha_k$ are additional parameters to be estimated.

\section{Issues Related to Nonparametric Density
  Estimation}\label{sec:noise} 
Throughout this section, $X_i$ and $T_i$ denote the original
duration and log-duration respectively. The duration variable $X$ is
measured in milliseconds and is hence discrete. It ranges from 1
millisecond up to more than 50000 milliseconds and exhibits an extremely
long tail. We explain how we deal with the long tail in
Section~\ref{app:tail}, and we explain how we deal with the
discreteness in Section~\ref{app:discrete}.
\subsection{The Tail of the Distribution}\label{app:tail}
Figure~\ref{fig:durationHist}
\begin{figure}[htb]
\begin{center}
\includegraphics[width = 1.1\textwidth, height = 2.5in]{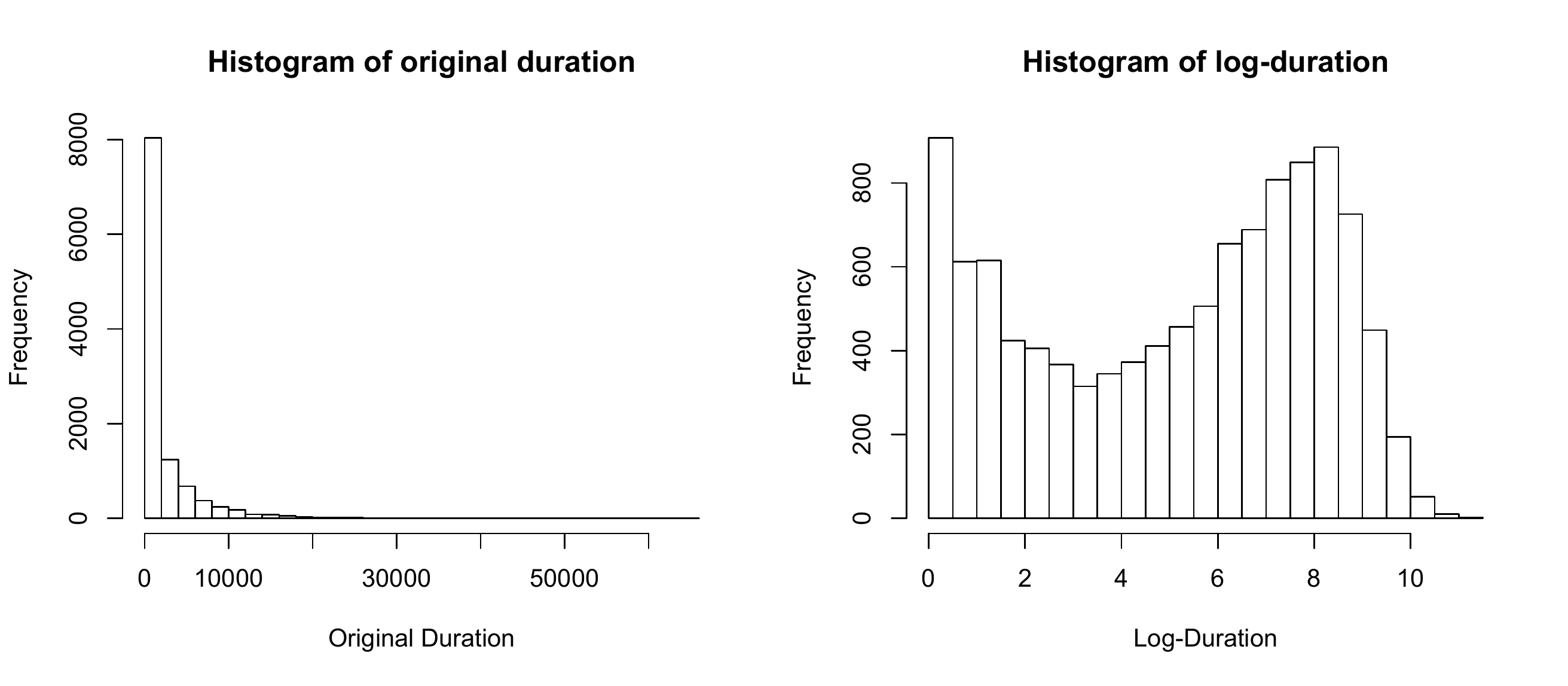}
\caption{Left panel: histogram of durations for one stock on one
  day. Right panel: histogram of log-durations on the same
  day.}\label{fig:durationHist} 
\end{center}
\end{figure}
shows histograms of durations
for one stock on one day both in the original scale and in the log scale.
The discreteness is merely an artifact of the measurement
process, and we would want to model duration as a continuous variable.  
The logarithm transformation is employed before modeling, due to the fact that
the long tail makes kernel conditional density
estimation with a single bandwidth unreliable.  For example, a
bandwidth small enough to avoid merging everything in the first bar of
the left panel in Figure~\ref{fig:durationHist} will be far too small for the
upper tail of the distribution.  Density estimation on
log-durations is equivalent to using a variable bandwidth on the
original scale.  For the reasons given above, we do all kernel density
estimation on the log scale for HSDM. Likewise in the semiparametric
variations of ACD, the long tail of the residuals $\epsilon$ also calls for
kernel density estimation (KDE) on the log scale. The two subfigures in 
Figure~\ref{fig:residualHist} 
\begin{figure}[htb]
\begin{center}
\begin{subfigure}{0.45\textwidth}
\caption{}
\includegraphics[width = \textwidth, height = 2.5in ]{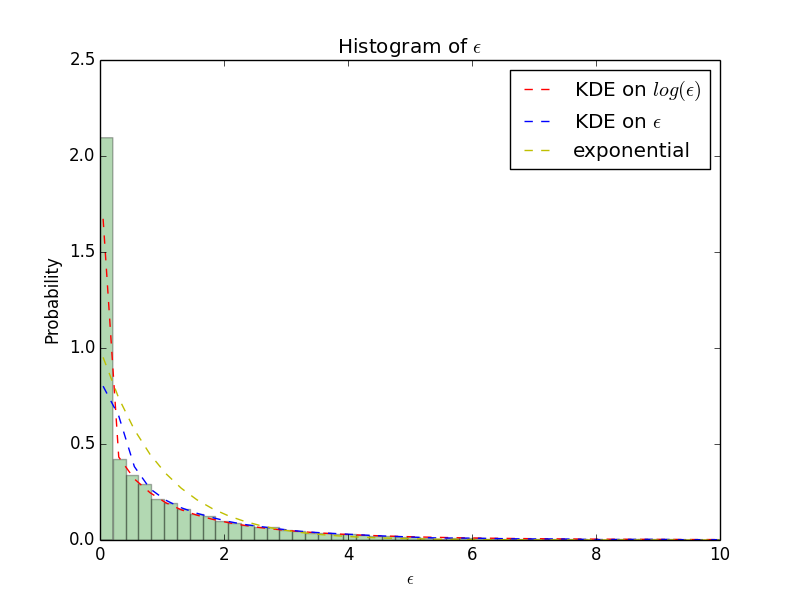}
\label{fig:epsHist}
\end{subfigure}
~
\begin{subfigure}{0.45\textwidth}
\caption{}
\includegraphics[width = \textwidth, height = 2.5in]{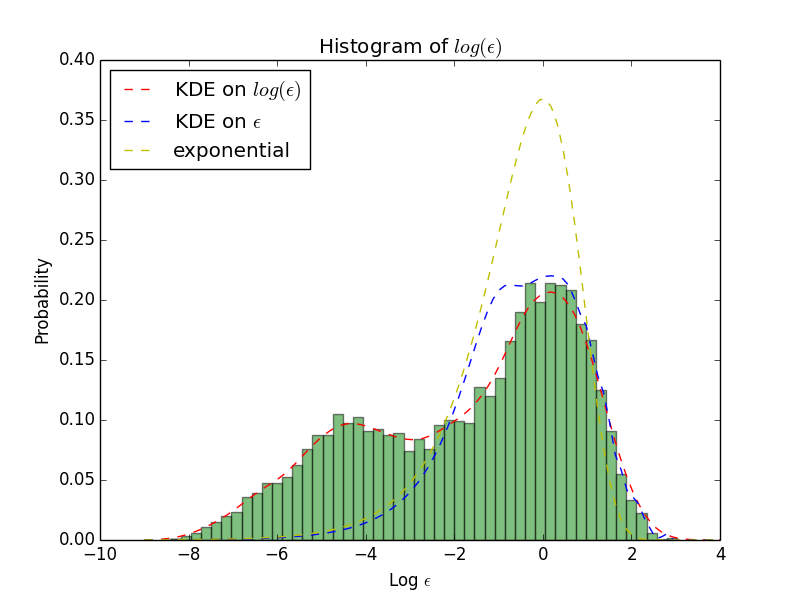}
\label{fig:logepsHist}
\end{subfigure}
\end{center}
\footnotesize
\caption{(a) Histogram of semiparametric ACD residuals. (b) Histogram
  of semiparametric ACD residuals in log scale. Each plot has
  superimposed the density of the standard exponential distribution
  along with two densities estimated by KDE: one based on the original
residuals, and the other based on log-residuals.}
\label{fig:residualHist}
\end{figure}
show the histograms of
FIACD residuals $\epsilon$ in the original scale and the log scale.
Superimposed on each histogram are
three fitted densities in the corresponding scales.  The three
densities are those of the standard exponential distribution,
a density estimated by KDE in the original scale, and one estimated by
KDE in the log scale. It is obvious that KDE
in the log scale captures the distribution of residuals more
accurately. Therefore,  kernel density estimation is applied to
$log(\epsilon)$ in semiparametric variants of ACD.

\subsection{Discreteness}\label{app:discrete}
Because the duration data are equally spaced (one-millisecond gaps), the
logarithms have larger gaps at low values and smaller gaps at large
values.  Since the durations are denser at low values, automatic
bandwidth selectors will tend to choose small bandwidths that put the
most common log-durations into separate bins while undersmoothing the
upper tail.  This is counterproductive, and
there is a simple method for obtaining more reasonable bandwidths.
We can replace each discretely-recorded duration $X_i$ by $X_i-U_i$,
where $U_i$ is an independently-generated noise value supported on an
interval of length one millisecond.  We choose $U_i$ to be uniform on
the interval $(0,1)$.  Let $Y_i=\log(1+X_i-U_i)$.  The reason for
adding the 1 before taking the logarithm is that $X_i=1$ corresponds to
$\log(X_i-U_i)<0$. If we used these data, kernel density estimation
would waste much of its effort estimating a density on $(-\infty,0)$,
while the data have no information to distinguish these values.
 So, we use the sequence
$\{Y_i\}_{i\geq1}$ to construct a conditional kernel density estimate
$g(\cdot | Y_{i-1})$, and then we convert this to an estimated
conditional density $\hat f_i$ for $T_i=\log[\exp(Y_i)-1]$ given
$T_{i-1}=t_{i-1}$ by 
\[\hat f_i(t)=g(\log[1+\exp(t)]|\log[1+\exp(t_{i-1})])
\frac{\exp(t)}{1+\exp(t)},\mbox{\ for $t>0$,}\]
for use in step~\ref{step1}  of the HSDM estimation.

Next, we give some justification for subtracting uniform random
variables before taking logarithms.  Subtracting a uniform
random variable from an integer-valued discrete random variable is
motivated by the discrete version of the probability integral
transform (PIT).  The well-known continuous version of the PIT is the following.

Let $\{Z_i\}_{i\geq1}$ be a sequence of random variables such that
$Z_1$ has CDF $F_1$, and the conditional CDF of $Z_i$ given
$Z_1=z_1,\ldots,Z_{i-1}=z_{i-1}$ is $F_i(\cdot|z_1,\ldots,z_{i-1})$
for $i>1$.  Assume that $F_i$ is a continuous CDF for all $i$. Define
$W_1=F_1(Z_1)$ and $W_i=F_i(Z_i|Z_1,\ldots,Z_{i-1})$ for $i>1$.  Then
$\{W_i\}_{i\geq1}$ are IID random variables with the uniform
distribution on the interval $(0,1)$.

The less well-known general version of the PIT is the following, of
which Proposition~\ref{pro:pit} is a corollary.
\begin{lemma}[General PIT]\label{lem:gpit}
Let $\{X_i\}_{i\geq1}$ be a sequence of random variables such that
$X_1$ has CDF $F_1$, and the conditional CDF of $X_i$ given
$X_1=x_1,\ldots,X_{i-1}=x_{i-1}$ is $F_i(\cdot|x_1,\ldots,x_{i-1})$
for $i>1$.  Define $F'_1(x)=\lim_{y\uparrow x}F_1(y)$, and for each
$i>1$ and each vector $(x_1,\ldots,x_{i-1},x)$ define
$F'_i(x|x_1,\ldots,x_{i-1})=\lim_{y\uparrow x}F_i(y|x_1,\ldots,x_{i-1})$.
Also, define $J_1(x)=F_1(x)-F'_1(x)$ and for $i>1$ define
$J_i(x|x_1,\ldots,x_{i-1})=F_i(x|x_1,\ldots,x_{i-1})-F'_i(x|x_1,\ldots,x_{i-1})$.
($J_i$ measures the sizes of any jumps in the CDF $F_i$.)
Let $\{V_i\}_{i\geq1}$ be a sequence of IID uniform random variables
on the interval $(0,1)$ that are independent of
$\{X_i\}_{i\geq1}$. Define $W_1=F_1(X_1)-(1-V_1)J_1(X_1)$ and for 
$i>1$ define 
\begin{equation}\label{eq:pit1}
W_i=F_i(X_i|X_1,\ldots,X_{i-1})-(1-V_i)J_i(X_i|X_1,\ldots,X_{i-1}).
\end{equation} 
Then $\{W_i\}_{i\geq1}$ are IID random variables with the uniform
distribution on the interval $(0,1)$.
\end{lemma} 
\begin{proof}
For each $i>1$ and each $0<p<1$, define
\[Q_i(p|x_1,\ldots,x_{i-1})=\inf\{x:F_i(x|x_1,\ldots,x_{i-1})\geq p\},\]
the generalization of the quantile function to general distributions,
which is continuous from the left on $(0,1)$. 
For $0<p<1$, $W_i\leq p$ if and only if
\begin{eqnarray*}
\mbox{either\ }X_i<Q_i(p|X_1,\ldots,X_{i-1})&\mbox{\ or\
}&\left(X_i=Q_i(p|X_1,\ldots,X_{i-1})\phantom{\frac{p}{J}}\right.\\
&&\left.\mbox{\ and\ }V_i\leq
\frac{p-F'(Q_i(p|X_1,\ldots,X_{i-1}))}{J_i(X_i|X_1,\ldots,X_{i-1})}\right).
\end{eqnarray*} 
Hence,
\begin{eqnarray*} 
\lefteqn{\Pr(W_i\leq p|W_1,\ldots,W_{i-1},X_1,\ldots,X_{i-1})}\\
&=&F'(Q_i(p|X_1,\ldots,X_{i-1}))\\
&&+J_i(p|X_1,\ldots,X_{i-1})
\frac{p-F'_i(Q_i(p|X_1,\ldots,X_{i-1}))}{J_i(X_i|X_1,\ldots,X_{i-1})}\\
&=&p,
\end{eqnarray*} 
and $W_i$ has the uniform distribution on $(0,1)$ conditional on
$W_1,\ldots,W_{i-1},X_1,\ldots,X_{i-1}$.  It follows that $W_i$ is
independent of  $W_1,\ldots,W_{i-1},X_1,\ldots,X_{i-1}$ and hence is
independent of $W_1,\ldots,W_{i-1}$.  The proof that $W_1$ has the
uniform distribution is essentially the same as above without the conditioning.
\end{proof} 

We now combine the two versions of the PIT to justify subtracting
independent uniform random variables from integer-valued random
variables.
\begin{lemma}\label{lem:both}
Assume the same conditions and notation as in  Lemma~\ref{lem:gpit},
and assume further that each $X_i$ is integer valued.
Let $Z_i=X_i-V_i$ for each $i$.  Let $G_1$ be the CDF of $Z_1$, and
let $G_i$ be the conditional CDF of $Z_i$ given $Z_1,\ldots,Z_{i-1}$
for $i>1$.  Each $G_i$ is a continuous CDF.  Also, $W_1=G_1(Z_1)$ and
$W_i=G_i(Z_i|Z_1,\ldots,Z_{i-1})$ 
for $i>1$, where $W_1$ and $W_i$ are defined in Lemma~\ref{lem:gpit}.
\end{lemma} 
\begin{proof} 
For each $i$, $X_i=\lfloor Z_i\rfloor+1$ and $V_i=X_i-Z_i$, so that
conditioning on the 
$Z_i$'s is equivalent to conditioning on the $X_i$'s and the $V_i$'s.
It is straightforward to see that each $G_i$ is the linear
interpolation of $F_i$ between consecutive integers.  That is
\begin{eqnarray} \label{eq:pit3}
&&\\ \nonumber
\lefteqn{G_i(z|z_1,\ldots,z_{i-1})}\\ \nonumber
&=&\bfun F_i(z|x_1,\ldots,x_{i-1})&\mbox{if
  $z$ is an integer,}\\
F_i(\lfloor z\rfloor|x_1,\ldots,x_{i-1})+(z-\lfloor
z\rfloor)J_i(\lfloor
z\rfloor+1|x_1,\ldots,x_{i-1})&\mbox{otherwise.}\efun
\end{eqnarray} 
It follows that 
\begin{eqnarray*} 
\lefteqn{G_i(Z_i|Z_1,\ldots,Z_{i-1})}\\
&=&F_i(X_i-1|X_1,\ldots,X_{i-1})
+(1-V_i)[F_i(X_i|X_1,\ldots,X_{i-1})\\
&&-F_i(X_i-1|X_1,\ldots,X_{i-1})]\\
&=&F_i(X_i|X_1,\ldots,X_{i-1})-(1-V_i)J_i(X_i|X_1,\ldots,X_{i-1})=W_i.
\end{eqnarray*} 
\end{proof} 

Lemma~\ref{lem:both} tells us that we can compute generalized
residuals two equivalent ways if we start with integer-valued random
variables.  One way is to use the general PIT in Lemma~\ref{lem:gpit}
directly on the integer-valued random variables.  If the
integer-valued random variables are the result of rounding up
unobserved continuous random variables, we might prefer to smooth out
the integers over the preceding interval and use the continuous PIT.
Lemma~\ref{lem:both} tells us that we get exactly the same sequence of
generalized residuals either way, so long as we use the same sequence
of uniform random variables for the smoothing as we use for
Lemma~\ref{lem:gpit}. 

There is one further connection between the smoothed and
integer-valued random variables.  They have the same prediction
log-likelihood.  
\begin{lemma}
Assume the same conditions as in Lemma~\ref{lem:both}.  Then for each
integer $x$ and each $z\in(x-1,x)$,
\begin{eqnarray*}
J_i(x|x_1,\ldots,x_{i-1})&=&\Pr(X_i=x|X_1=x_1,\ldots,X_{i-1}=x_{i-1})\\
&=&g_i(z|x_1,\ldots,x_n),
\end{eqnarray*} 
where $g_i$  is the conditional density of $Z_i$ given
$Z_1=z_1,\ldots,Z_{i-1}=z_{i-1}$.
\end{lemma} 
\begin{proof}
The proof is straightforward from (\ref{eq:pit3}), since $G_i$ is
piecewise linear and the slope of
$G_i$ on the open interval of length 1 to the left of each integer
$x$ is  $J_i(x|x_1,\ldots,x_{i-1})$. 
\end{proof} 

Of course, all of the previous discussion relies on knowing all of the
conditional CDFs of the integer-valued random variables.  In our
modeling of durations, we must estimate these distributions.  Since we
think of the durations as having continuous distributions in an ideal
system, we choose the smoothing method.  First, we smooth the
integer-valued durations, then we estimate the conditional distributions.
Our estimated distributions will not be piecewise linear as are the
$G_i$ in Lemma~\ref{lem:both}.  Since both the training data and test
data are integer-valued, we convert our conditional distributions back
to either discrete or smoothed distributions whichever makes the
desired calculations simpler.
  
\subsection{Random Effect of Smoothing}
To better understand the extent to which smoothing the integer-valued
durations  affects our 
results, we ran experiments in which we repeated the smoothing using
several independent sequences of uniform random variables (the $V_i$'s.) 
We found that the
construction of $\hat f_i$ in step~\ref{step1} and the fit of the
ARFIMA model in step~\ref{step3} produced indistinguishable results
from one smoothing to the next.  Also, the final generalized residuals
behave the same with regard to the various diagnostics and the prediction
log-likelihood from one smoothing to the next. 

Figure~\ref{fig:noiseRobust}\begin{figure}[htb]
\begin{center}
\includegraphics[width = 5.5in, height = 3.0in]{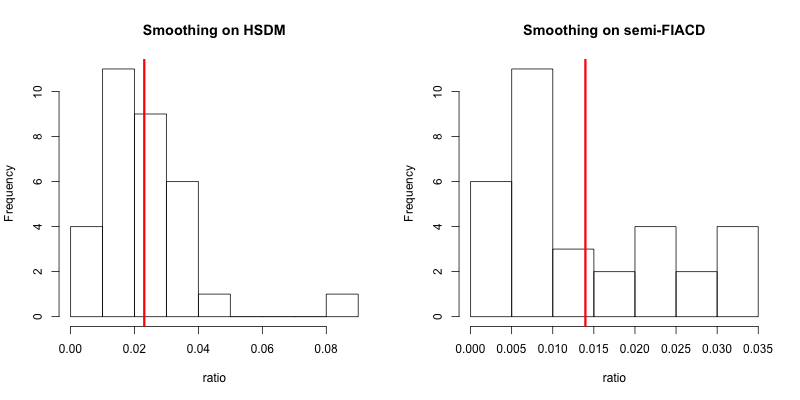}
\caption{ Histograms of ratio (\ref{eq:ratio}). One ratio is
  calculated for each of the 32 stock/date pairs.
The Left subfigure is for HSDM, and the right
  subfigure is for semiparametric-FIACD.}
\label{fig:noiseRobust}
\end{center}
\end{figure}
shows that the improvement of the HSDM model
over the best benchmark model (semiparametric-FIACD) is robust in the
sense of prediction log-likelihood under multiple smoothings of the
integer-valued observations.  We did three smoothings on each of 32
stock/day pairs. The left
subfigure is the histogram of the range of the three different
prediction log-likelihoods from different smoothings as a fraction of the
amount by which the HSDM prediction log-likelihood exceeds that of the
semiparametric-FIACD, i.e. 
\begin{equation}\label{eq:ratio}
{\rm ratio} = \frac{ \max\{ll_1, ll_2, ll_3\}-\min\{ll_1,ll_2,ll_3\} }
{ ll_{\rm HSDM} - ll_{\rm semiparametric-FIACD}},
\end{equation} 
where $ll_i$ is the log-likelihood of HSDM under the $i$th different
smoothing for $i=1,2,3$. The vertical red line is the average. In summary, the
deviation caused by smoothing is only 2 percent of the improvement in prediction
log-likelihood.  
Similarly, the right subfigure shows the distribution of ratio in
(\ref{eq:ratio})  when
smoothing is applied to the semiparametric-FIACD model. The
effect of smoothing on the semiparametric-FIACD model is even smaller. Although
the semiparametric-FIACD model 
does not require smoothing, doing the same smoothing as we do in the HSDM
model gives a fairer comparison between the two models.

\bibliographystyle{imsart-nameyear}
\bibliography{HSDM}

\end{document}